\documentclass[11pt,leqno]{article}
\usepackage{amsmath,amssymb,amsfonts}
\newtheorem{thm}{Theorem}[section]
\newtheorem{dfn}{Definition}[section]
\newtheorem{ex}{Example}[section]
\newtheorem{remark}{Remark}[section]
\newtheorem{proposition}{Proposition}[section]
\newtheorem{corollary}{Corollary}[section]
\newtheorem{claim}{Claim}[section]

\newenvironment{proof}{\mbox{\bf Proof.}}{\mbox{$\dashv$}\bigskip}
\begin{document}
\begin{center}
{\Large\bf Disappointment in Social Choice Protocols}\\
\vspace{.25in}
{ Mohammad Ali Javidian, Rasoul Ramezanian\footnote{Corresponding Author}}\\
 {Department of Mathematics Sharif University of Technology, Tehran, Iran}\\
{Complex and Multi Agent Systems Lab}\\
ali.javidian@gmail.com, ramezanian@sharif.edu
\end{center}
\begin{abstract}
\noindent Social choice theory is a theoretical framework for analysis of combining individual preferences, interests, or welfares to reach a collective decision or social welfare in some sense. We introduce a new criterion for social choice protocols called "social disappointment". Social disappointment happens when the outcome of a voting system occurs for those alternatives which are at the end of at least half of individual preference profiles. Here we introduce some protocols that prevent social disappointment and prove an impossibility theorem based on this key concept. \vspace*{1 cm}

\end{abstract}
\section{Introduction}

In social sciences, we are facing two kinds of social choices: voting which is used to make a political decision and market mechanism as a tool to make an economic decision\cite{A1}. Here, we are merely concerned with the former.

The theory underlying voting systems is known as social choice theory and is concerned with the design and analysis of methods for collective decision making\cite{Suz}. Voting procedures are among the most important methods for collective decision
making. In this paper, our attention is on voting procedures. Voting procedures
focus on the aggregation of individuals' preferences to produce collective decisions.
In practice, a voting procedure is characterized by ballot responses and the way ballots are tallied to determine winners. Voters are assumed to have clear preferences
over candidates (alternatives) and attempt to maximize satisfaction with the election outcome by their ballot responses. Voting procedures are formalized by social
choice functions, which map ballot response profiles into election outcomes(see\cite{BF}, page:175).

We use a broad class of social choice functions such as Condorcet method, Plurality rule, Hare system, Borda count, Sequential Pairwise Voting with a Fixed
Agenda (Seq. Pairs), and Dictatorship. Condorcet method is typically attributed
to the Marquis de Condorcet (1743-1794); However, it dates back to Ramon Llull
in the thirteenth century([28] p. 6). Hare procedure was introduced by Thomas
Hare in 1861, and is also known by names such as the "single transferable vote system"(STV) or "instant runoff voting" (\cite{Taylor} p. 7). Jean Charles Chevalier de Borda
(1733-99) in 1781 \cite{Borda} introduced an aggregation procedure known as Borda count.
Interestingly, recent historical work by McLean, Urken (1993) \cite{Mc}, and Pukelsheim
(unpublished) reveals that Bordas system had been explicitly described in 1433 by
Nicholas of Cusa (1401-64), a Renaissance scholar interested in the question of how
German kings should be elected (\cite{Taylor2} p. 9). For more details and examples, see
\cite{Taylor} sec. 1.3.

There are five desirable properties that relate to such procedures: Always-A-Winner Condition (AAW), Condorcet Winner Criterion (CWC), Pareto Condition,
Monotonicity (Mono), Independence of Irrelevant Alternatives (IIA). A social choice
procedure is said to satisfy AAWcondition if every sequence of individual preference
lists produces at least one winner. An alternative 'x' is said to be a Condorcet winner
if it is the unique winner in Condorcets method. A social choice procedure is said to
satisfy CWC provided that-if there is a Condorcet winner-then it alone is the social
choice. A social choice procedure is said to satisfy the Pareto condition (or just
Pareto) if the following holds for every pair 'x' and 'y' of alternatives: If everyone
prefers 'x' to 'y', then y is not a social choice. A social choice procedure is said
to be monotone provided that the following holds for every alternative 'x': If 'x' is
the social choice (or tied for such) and someone changes his or her preference list
by moving 'x' up one spot, then 'x' should still be the social choice (or tied for
such). A social choice procedure has IIA condition the social choice set includes 'x'
but not 'y', and one or more voters change their preferences, but no one changes
his or her mind about whether 'x' is preferred to 'y' or 'y' to 'x', then the social
choice set should not change so as to include 'y'. The condition of "independence
of irrelevant alternatives" was first used by Arrow \cite{A1} in 1951. For more details of
these properties see \cite{Taylor}, sec. 1.4.

In this paper, we propose a new property for social choice procedures called \textbf{social disappointment in voting} which is a situation that happens when the winner is the least favorable candidate for at least half of the voters. See the definition and an example of social disappointment in section 2.

The rest of the paper is organized as follows. In section 2 we introduce the concept of social disappointment and the Least Public Resentment procedure (L.P.R) which is a procedure that prevents social disappointment in voting. In section 3 in imitation of Taylor's work (see\cite{Taylor3} also\cite{Taylor} pp. 28-31) we present an impossibility theorem based on social disappointment in voting.

\textbf{Notice}: In this paper, we will follow the notation and basic results of Taylor and Pacelli(see\cite{Taylor}, chapter:1).


\section{Social Disappointment}\label{sec1}

We start explaining the social disappointment property by the following example.
\begin{ex}
Consider the following situation in which there are four Dutchmen, three Germans, and two Frenchmen who have to decide which drink will be served for lunch (only a single drink will be served to all).\begin{center}
\newcommand{\head}[1]{\textnormal{\text{#1}}}
\begin{table}[!hbp]
\centering
\begin{tabular}{ccc}
\head{4} & \head{3} & \head{2} \\
\hline
Milk& Beer& Wine \\
Wine & Wine& Beer \\
Beer& Milk & Milk
\end{tabular}
\end{table}
\end{center}
Now, which drink should be served based on these individual preferences? Milk
could be chosen since it has the most agents ranking it first. Milk is the winner
according to the plurality rule, which only considers how often each alternative is
ranked in first place. However, a majority of agents will be dissatisfied with this
choice as they prefer any other drink to Milk(see\cite{BCE}, pp. 3,4).
\end{ex}
\begin{dfn}
\textbf{Social disappointment in voting} happens when the outcome of a voting system (for 3 or more alternatives) occurs for those alternatives which are at the end of at least half of individual preference profiles.
\end{dfn}
Now the question is for which protocols social disappointment(S.D) may hap-
pen? The answer is given in the following table (where a 'yes' indicates that social
disappointment may happen).
\begin{center}
\newcommand{\head}[1]{\textnormal{\text{#1}}}
\begin{table}[!hbp]
\centering
\begin{tabular}{c|c|c|c|c|c|c}
\head{} & \head{Plurality} & \head{Borda} & \head{Hare} & \head{Seq. Pairs} & \head{Dictator} & \head{Condorcet} \\
\hline
S.D & Yes & Yes & Yes & Yes & Yes & Yes
\end{tabular}
\end{table}
\end{center}
We prove the seven claims in Table 1.
\begin{claim}
The Plurality rule does not prevent social disappointment.
\end{claim}
\begin{proof}
See Example 2.1.
\end{proof}
\begin{claim}
The Borda count does not prevent social disappointment.
\end{claim}
\begin{proof}
Consider the three alternatives 'a', 'b', and 'c' and the following sequence of two preference lists:
\begin{center}
\newcommand{\head}[1]{\textnormal{\text{#1}}}
\begin{table}[!hbp]
\centering
\begin{tabular}{cc}
Voters & Voters \\
1and2 & 3and4 \\
\hline
a& c \\
b& b \\
c & a
\end{tabular}
\end{table}
\end{center}
Alternatives 'a', 'b' and 'c' are the social choice when the Borda count procedure is used. Although 'a' is the social choice (also 'c'), but it is at the bottom of individual preference lists and so social disappointment has taken place.
\end{proof}
\begin{claim}
The Hare procedure does not prevent social disappointment.
\end{claim}
\begin{proof}
Consider the three alternatives 'a', 'b', and 'c' and the following sequence of ten preference lists grouped into voting blocks of size four, three, and two:
\begin{center}
\newcommand{\head}[1]{\textnormal{\text{#1}}}
\begin{table}[!]
\centering
\begin{tabular}{ccc}
Voters & Voters& Voters \\
1-4 & 5-7 & 8-10 \\
\hline
a& c & b \\
b& b & c \\
c & a & a
\end{tabular}
\end{table}
\end{center}
Alternative 'a' is the social choice when the Hare system is used. Although 'a' is the social choice, but it is at the bottom of individual preference lists and so social disappointment has taken place.
\end{proof}
\begin{claim}
Sequential pairwise voting with a fixed agenda does not prevent social disappointment.
\end{claim}
\begin{proof}
Consider the three alternatives 'b', 'c', and a and suppose that this ordering of the alternatives is also the agenda. Consider the following sequence of four preference lists grouped into voting blocks of size two, one, and one:
\begin{center}
\newcommand{\head}[1]{\textnormal{\text{#1}}}
\begin{table}[!hbp]
\centering
\begin{tabular}{ccc}
Voters & Voter& Voter \\
1and2 & 3 & 4 \\
\hline
a& b & c \\
b& c & b \\
c & a & a
\end{tabular}
\end{table}
\end{center}
Alternatives 'a', 'b' are the social choice when the Hare system is used. Although 'a' is the social choice, but it is at the bottom of individual preference lists and so social disappointment has taken place.
\end{proof}
\vskip 1cm
\begin{claim}
A dictatorship does not prevent social disappointment.
\end{claim}
\begin{proof}
Consider the three alternatives 'a', 'b', and 'c' and the following three preference lists:
\begin{center}
\newcommand{\head}[1]{\textnormal{\text{#1}}}
\begin{table}[!hbp]
\centering
\begin{tabular}{ccc}
Voters & Voter& Voter \\
1and2 & 3 & 4 \\
\hline
a& c & c \\
b& b & b \\
c & a & a
\end{tabular}
\end{table}
\end{center}
Assume that Voter 1 is the dictator. Then 'a' is the social choice, but obviously social disappointment has happened.
\end{proof}
\begin{claim}
If in Condorcet method more than half of voters put 'a' at the bottom
of individual preference lists then for sure 'a' would not be the social choice and
in this case social disappointment would not occur. But if the number of voters is
even and precisely half of voters put 'a' in the end of their lists, one of these two
possibilities happens:
\begin{itemize}
\item 1. Not all voters in the other half put 'a' at the top of their lists, which in this
case, 'a' definitely does not hold the social choice and social disappointment
occures.
\end{itemize}
\begin{itemize}
\item 2. All the voters in the other half also put 'a' at the top of their lists, which in this case 'a' is definitely  in the set of social choice and therefore we face with the social disappointment.
\end{itemize}
\end{claim}
\begin{proof}
It is concluded from the definitions.
\end{proof}
\begin{remark}
Regarding case 2 in claim 2.6 if we have only three alternatives, the set of social choice certainly have more than one member.
\end{remark}
\subsection{The least public resentment procedure (L.P.R):}
We observed that none of the famous procedures listed in Table 1 prevents social disappointment.
We introduce a procedure which prevents social disappointment in voting, called The least public resentment procedure (L.P.R). In this procedure
we begin by deleting the alternative or alternatives occurring at bottom of the most lists. At this stage we have lists that are at least one alternative shorter than the lists we started with. Now, we simply repeat this procedure of deleting the least public resentment alternative or alternatives. The alternative(s) deleted last is declared as the winner.
\begin{ex}
Consider Example 2.1, we decide which alternative occurs at the bottom of the most lists and delete it from all the lists. Milk is deleted from each list leaving the following:
\begin{center}
\newcommand{\head}[1]{\textnormal{\text{#1}}}
\begin{table}[!hbp]
\centering
\begin{tabular}{ccc}
\head{4} & \head{3} & \head{2} \\
\hline
Wine& Beer& Wine \\
Beer&Wine & Beer
\end{tabular}
\end{table}
\end{center}
Now, Beer occurs at the bottom of six of the nine lists, and thus is eliminated. Hence, Wine is the social choice when the L.P.R is used.
\end{ex}
Which properties does this procedure satisfy? The answer is given in the following table:
\begin{center}
\newcommand{\head}[1]{\textnormal{\text{#1}}}
\begin{table}
\centering
\begin{tabular}{c|c|c|c|c|c|c}
\head{} & \head{AAW} & \head{CWC} & \head{Pareto} & \head{Mono}& \head{IIA} & \head{Non S.D}\\
\hline
L.P.R & Yes & No & Yes & No & No & Yes
\end{tabular}
\end{table}
\end{center}
\begin{proposition}
The L.P.R procedure satisfies AAW, Pareto, and nonexistence of social disappointment (Non S.D) criterion but does not satisfy CWC, Monotonicity, and IIA.
\end{proposition}
\begin{proof}
For this procedure, the description makes it clear that there is at least one winner for every profile. So L.P.R satisfies AAW condition.

Suppose that there is a winner that is in the end of at least half of preference
profiles. The L.P.R procedure would delete this alternative from profile lists in the
first stage, so the social disappointment for this alternative could not occure.

Consider the three alternatives 'a', 'b', and 'c' and the following sequence of seven preference lists grouped into voting blocks of size two, two, two, and one:
\begin{center}
\newcommand{\head}[1]{\textnormal{\text{#1}}}
\begin{table}[!hbp]
\centering
\begin{tabular}{cccc}
Voters & Voters & Voters &Voter \\
1and2 & 3and4 & 5and6 & 7 \\
\hline
a&a &b &c \\
b&c &c & b\\
c&b &a &a
\end{tabular}
\end{table}
\end{center}
The L.P.R procedure produces 'b' as the social choice. However, 'a' is clearly the Condorcet winner, defeating each of the other alternatives in one-on-one competitions. Since the Condorcet winner is not the social choice in this situation, we have that the L.P.R procedure does not satisfy the Condorcet winner criterion.

L.P.R procedure satisfies the Pareto condition. Because if in all lists 'b' has occurred down below 'a', therefore at some point 'b' would be gone but 'a' stands still due to the fact that 'b' socially is more resentful than 'a'. So based on social choice procedure in this system 'b' would be eliminated in early stages or at most in comparison with 'a'.

Consider the three alternatives 'a', 'b', and 'c' and the following sequence of seven preference lists grouped into voting blocks of size two, two, two, and one:
\begin{center}
\newcommand{\head}[1]{\textnormal{\text{#1}}}
\begin{table}[!hbp]
\centering
\begin{tabular}{ccc}
Voters & Voters & Voter \\
1and2 & 3and4 & 5\\
\hline
b & c & b \\
a & a & c \\
c & b & a
\end{tabular}
\end{table}
\end{center}
We delete the alternatives which have taken place more than the other alternatives at the end of the votes. In this case, that would be alternatives 'c' and 'b' with the two last places in votes for each as compared to one for 'a'. But now 'a' is the only alternative left, and so it is the social choice when the L.P.R procedure is used.

Now suppose that the single voter on the most right changes his or her list by interchanging 'a' with the alternative that is right above 'a' on this list. This apparently favorable-to-'a'-change yields the following sequence of preference lists:
\begin{center}
\newcommand{\head}[1]{\textnormal{\text{#1}}}
\begin{table}[!hbp]
\centering
\begin{tabular}{ccc}
Voters & Voters & Voter \\
1and2 & 3and4 & 5\\
\hline
b & c & b \\
a & a & a \\
c & b & c
\end{tabular}
\end{table}
\end{center}
If we apply the L.P.R procedure again, we delete the alternatives which have taken place more than the other alternatives at the end of the votes. In this case, 'c' is that alternative. But the reader can now easily check that with 'c' so eliminated, alternative 'b' is at bottom of two of the fifth lists. Alternative 'a' is deleted and so 'b' is the social choice. This change in social choice from 'a' to 'c' shows that the L.P.R procedure does not satisfy monotonicity.

Consider the three alternatives 'a', 'b', and 'c' and the following sequence of four preference lists grouped into voting blocks of size one, one, and two:
\begin{center}
\newcommand{\head}[1]{\textnormal{\text{#1}}}
\begin{table}[!hbp]
\centering
\begin{tabular}{ccc}
Voter & Voter & Voters \\
1 & 2 & 3and4 \\
\hline
b & a & b \\
a & c & c \\
c & b & a
\end{tabular}
\end{table}
\end{center}
Alternative 'b' is the social choice when the L.P.R procedure is used. In particular, 'b' is a winner and 'a' is a non-winner. Now suppose that Voter 4 changes his or her list by interchanging the alternatives 'a' and 'c'. The lists then become:
\begin{center}
\newcommand{\head}[1]{\textnormal{\text{#1}}}
\begin{table}[!hbp]
\centering
\begin{tabular}{cccc}
Voter1 & Voter2 & Voter3 & Voter4 \\
\hline
b & a & b & b \\
a & c & c & a \\
c & b & a & c
\end{tabular}
\end{table}
\end{center}
Notice that we still have 'b' over 'a' in Voter 4’s list. However, L.P.R procedure now has 'a' and 'b' tied for the win. Thus, although no one changed his or her mind about whether 'a' is preferred to 'b' or 'b' to 'a', the alternative 'a' went from being a non-winner to being a winner. This shows that independence of irrelevant alternatives fails for the L.P.R procedure.
\end{proof}
\section{A Glimpse of Impossibility}\label{sec2}
Taylor proved in \cite{Taylor3} also \cite{Taylor} pp. 28-31 that there is no social choice procedure for three or more alternatives that satisfies the always-a-winner criterion, independence of irrelevant alternatives, and the Condorcet winner criterion. Now we prove an impossibility theorem based on social disappointment concept.
\begin{thm}
There is no social choice procedure for four or more alternatives that satisfies the nonexistence of S.D. criterion, and the Condorcet winner criterion.
\end{thm}
\begin{proof}
We assume that we have a social choice procedure that satisfies the Condorcet winner criterion. We then show that if this procedure is applied to the profile that consists of Condorcet’s voting paradox \cite{C}, then it produces a winner which will lead to social disappointment. We prove this claim for when we have four alternatives.

Assume that we have a social choice procedure that satisfies the Condorcet winner criterion. Consider the following profile:
\begin{center}
\newcommand{\head}[1]{\textnormal{\text{#1}}}
\begin{table}[!hbp]
\centering
\begin{tabular}{cccccc}
\hline
d & d & d & c & b & b\\
a & a & c & a & c & c \\
b & b & a & b & a & a \\
c & c & b & d & d & d
\end{tabular}
\end{table}
\end{center}
Alternative 'd' is the unique social choice when the Condorcet’s method is used. Although 'd' is the social choice, but it is at the bottom of individual preference lists and so social disappointment has taken place.
\end{proof}
\begin{corollary}
There is no social choice procedure for four or more alternatives that satisfies:
\begin{itemize}
\item a) the nonexistence of S.D. criterion, always a winner, and the Condorcet winner criterion.
\end{itemize}
\begin{itemize}
\item b) the nonexistence of S.D. criterion, monotonicity, and the Condorcet winner criterion.
\end{itemize}
\begin{itemize}
\item c) the nonexistence of S.D. criterion, Independence of Irrelevant Alternatives, and the Condorcet winner criterion.
\end{itemize}
\begin{itemize}
\item d) the nonexistence of S.D. criterion, Pareto, and the Condorcet winner criterion.
\end{itemize}
\end{corollary}
\begin{proof}
It is obvious considering Theorem 3.1.
\end{proof}
\subsection{Condorcet with an amendment and Seq. Pairs with an amendment procedure}
Before considering the rest of the possible cases, we introduce following two procedures and investigate their properties which we mentioned in this article.

\noindent
\textbf{Condorcet with an amendment procedure}

\noindent
This protocol is done the same as Condorcet’s method, with the difference that in the end we remove those alternatives with the social disappointment from the set of social choice.

\noindent
\textbf{Seq. Pairs with an amendment procedure}

\noindent
This protocol is done the same as Seq. Pairs, with the difference that in the end we remove those alternatives with the social disappointment from the set of social choice.

Which properties do the procedures satisfy? The answer is given in the following table:
\begin{center}
\newcommand{\head}[1]{\textnormal{\text{#1}}}
\begin{table}[!hbp]
\centering
\begin{tabular}{|c|c|c|c|c|c|c|}
\hline
\head{} & \head{AAW} & \head{CWC} & \head{Pareto} & \head{Mono}& \head{IIA} & \head{Non S.D}\\
\hline
Condorcet with an amendment & No & No & Yes & Yes & Yes & Yes\\
\hline
Seq. Pairs with an amendment & Yes & No & Yes & Yes & No & Yes\\
\hline
\end{tabular}
\end{table}
\end{center}
Here we will only prove three items of the table above. We leave the rest for the reader (The proof will be easy to comprehend according to the given definitions and \cite{Taylor} sections 1.5,1.6 ).
\begin{proposition}
Condorcet with an amendment and Seq. Pairs with an amendment procedures do not satisfy the Condorcet winner criterion. Furthermore Condorcet with an amendment procedure does not satisfy the always a winner criterion.
\end{proposition}
\begin{proof}
Consider the following profile:
\begin{center}
\newcommand{\head}[1]{\textnormal{\text{#1}}}
\begin{table}[!hbp]
\centering
\begin{tabular}{cccccc}
\hline
d & d & d & c & b & b\\
a & a & c & a & c & c \\
b & b & a & b & a & a \\
c & c & b & d & d & d
\end{tabular}
\end{table}
\end{center}

Alternative 'd' is the unique social choice when the Condorcet’s method is used. Therefore there is no winner (NW) when the Condorcet with an amendment procedure is used. Consider Seq. Pairs voting with fixed agenda a,b,c,d. Alternatives 'c', 'd' are the social choices when the Seq. Pairs voting with this particular fixed agenda is used. Thus alternative 'c' is the social choice when Seq. Pairs with an amendment is used, so Condorcet with an amendment and Seq. Pairs with an amendment procedures do not satisfy the Condorcet winner criterion. Furthermore according to what was said Condorcet with an amendment procedure does not satisfy the always a winner criterion.

\end{proof}
\begin{remark}
In light of Remark 2.1 we understand that set of social choice would be either all three alternatives or two alternatives which one would be free from social disappointment. Anyway, set of social choice would include the alternative with no social disappointment. Considering this fact the following proposition will be prove.
\end{remark}
\begin{proposition}
There exist some social choice procedures for three alternatives that satisfy:
\begin{itemize}
\item a) the nonexistence of S.D. criterion, and the Condorcet winner criterion.
\end{itemize}
\begin{itemize}
\item b) the nonexistence of S.D. criterion, the Condorcet winner, and always a winner criterion.
\end{itemize}
\begin{itemize}
\item c) the nonexistence of S.D. criterion, the Condorcet winner, and Pareto criterion.
\end{itemize}
\begin{itemize}
\item d) the nonexistence of S.D. criterion, the Condorcet winner, and monotonicity criterion.
\end{itemize}
\begin{itemize}
\item e) the nonexistence of S.D. criterion, the Condorcet winner, and independence of irrelevant alternatives criterion.
\end{itemize}
\begin{itemize}
\item f) the Pareto criterion, the Condorcet winner, and independence of irrelevant alternatives criterion.
\end{itemize}
\end{proposition}
\begin{itemize}
\item g) the nonexistence of S.D. criterion, monotonicity, and independence of irrelevant alternatives criterion.
\end{itemize}
\begin{proof}
For (b) consider Seq. Pairs with an amendment and for the rest Condorcet's method with an amendment.
\end{proof}
\newpage
\subsection{The Least Unpopular (L.U) and The Least Unpopular Reselection procedure(L.U.R)}
To investigate the remaining cases we introduce and investigate The Least Unpopular and The Least Unpopular Reselection procedures.

\noindent
\textbf{The Least Unpopular procedure (L.U)}

\noindent
The social choice in this procedure is (are) the alternative(s) that appears lesser than the others at the bottom of individual preference lists. This protocol satisfies AAW, monotonicity, and nonexistence of S.D. criterion, but does not satisfy CWC, Pareto, and IIA criterion.
\begin{proposition}
The Least Unpopular procedure does not satisfy CWC, IIA, and Pareto criterion.
\end{proposition}
\begin{proof}
Consider the four alternatives 'a', 'b', 'c', and 'd' and the following profile:
\begin{center}
\newcommand{\head}[1]{\textnormal{\text{#1}}}
\begin{table}[!hbp]
\centering
\begin{tabular}{ccc}
Voters & Voter & Voter \\
1and2 & 3 & 4 \\
\hline
a & c & d \\
b & a & a \\
c & b & b \\
d & d & c
\end{tabular}
\end{table}
\end{center}
Alternatives 'a', 'b' are the social choices when the Least Unpopular procedure is used. Thus, alternative 'b' is in the set of social choice even though everyone prefers 'a' to 'b'. This show that Pareto fails.
Now consider the three alternatives 'a', 'b', 'c' and following profile:
\begin{center}
\newcommand{\head}[1]{\textnormal{\text{#1}}}
\begin{table}[!hbp]
\centering
\begin{tabular}{cc}
Voters & Voter \\
1and2 &3 \\
\hline
a & b \\
b & c \\
c & a
\end{tabular}
\end{table}
\end{center}
Alternative 'b' is the social choice when the Least Unpopular procedure is used. However, 'a' is clearly the Condorcet winner, defeating each of the other alternatives in one-on-one competitions. Since the Condorcet winner is not the social choice in this situation, we have that the Least Unpopular procedure does not satisfy the Condorcet winner criterion. On the other hand
'b' is a non-winner. Now suppose that voter 3 changes his or her list by interchanging the alternatives 'a' and 'c'.
The lists then become:
\begin{center}
\newcommand{\head}[1]{\textnormal{\text{#1}}}
\begin{table}
\centering
\begin{tabular}{cc}
Voters & Voter \\
1and2 & 3 \\
\hline
a & b \\
b & a \\
c & c
\end{tabular}
\end{table}
\end{center}
Notice that we still have 'b' over 'a' in Voter 4’s list. However, Least Unpopular procedure now has 'a' and 'b' tied for the win. Thus, although no one changed his or her mind about whether 'a' is preferred to 'b' or 'b' to 'a', the alternative 'a' went from being a non-winner to being a winner. This shows that independence of irrelevant alternatives fails for the Least Unpopular procedure.
\end{proof}
\begin{proposition}
There are some social choice procedures for three or more alternatives that satisfy:

a) the nonexistence of S.D. criterion, monotonicity, and always a winner criterion.

b) the nonexistence of S.D. criterion, Pareto, and always a winner criterion.
\end{proposition}
\begin{proof}
The Least Unpopular procedure is one of them.
\end{proof}
Now we introduce and investigate The Least Unpopular Reselection procedure.

\noindent
\textbf{The Least Unpopular Reselection (L.U.R)}

\noindent
First of all we choose a set of alternatives appearing lesser than the others at the bottom of individual preference lists. If this set has only one member, it would be the social choice. Otherwise we remove the remaining alternatives and run the L.U. procedure for the set obtained from the first stage, and keep doing this until there will be no continuing. The obtained set in the last repetition would be the social choice. This protocol satisfies AAW, monotonicity, Pareto, and nonexistence of S.D. criterion, but does not satisfy CWC, and IIA criterion.
\begin{proposition}
The Least Unpopular Reselection procedure does not satisfy CWC and IIA criterion.
\end{proposition}
\begin{proof}
Perform as we did in the proof of Proposition 3.3.
\end{proof}
\begin{proposition}
There are some social choice procedures for three or more alternatives that satisfy the nonexistence of S.D. criterion, Pareto, and monotonicity criterion.
\end{proposition}
\begin{proof}
The L.U.R procedure is one of them.
\end{proof}
\vspace*{.5cm}

\section{Conclusion and Future Direction}\label{sec3}

Here in the following table we summarize what we came to conclusion in the previous sections:
\begin{center}
\newcommand{\head}[1]{\textnormal{\text{#1}}}
\begin{table}
\centering
\begin{tabular}{|l|c|c|c|c|c|c|}
\hline
\head{} & \head{AAW} & \head{CWC} & \head{Pareto} & \head{Mono}& \head{IIA} & \head{Non S.D}\\
\hline
Condorcet & no & yes & yes & yes & yes & no \\
\hline
Plurality & yes & no & yes & yes & no & no \\
\hline
Borda count & yes & no & yes & yes & no & no \\
\hline
Hare system & yes & no & yes & no & no & no \\
\hline
Seq. Pairs & yes & no & no & yes & no & no \\
\hline
Dictator & yes & no & yes & yes & no & no \\
\hline
L.P.R & yes & no & yes & yes & no & yes \\
\hline
Condorcet with an amendment & no & no & yes & yes & yes & yes \\
\hline
Seq. Pairs with an amendment & yes & no & yes & yes & no & yes \\
\hline
L.U & yes & no & no & yes & no & yes \\
\hline
L.U.R & yes & no & yes & yes & no & yes \\
\hline
\end{tabular}
\end{table}
\end{center}
Finally we must note that the well-known procedures that satisfy the condition IIA are Condorcet extension which means that they choose the Condorcet winner whenever one exists. Since CWC is not compatible with the nonexistence of S.D. criterion, we don’t know whether there are some social choice procedures for three alternatives that satisfy the nonexistence of S.D. criterion, always a winner criterion, and independence of irrelevant alternatives criterion.

This question can be a motivation for future work.
\section*{History and Related Works}
Kenneth Arrow showed that it is impossible to design a voting rule that satisfies some very natural properties (Arrow, 1950)\cite{A2}. This seminal work is thus named Arrow’s impossibility theorem, and is broadly regarded as the beginning of modern Social Choice Theory, which is an active research direction in Economics\cite{X}.

In terms of the literature, there are, as one might expect, an abundance of treatments of Arrow’s theorem. Two of the most notable book-length treatments are Arrow (1963)\cite{A2} and Kelly (1978)\cite{Kelly}. Textbook coverage (with proofs) is also widely available, including chapters in Kelly (1987)\cite{Kelly2}, Saari (1995)\cite{Saari}, and Taylor(1995)\cite{Taylor}.

The 1970s seem to have been the heyday for social choice research in the second half of the twentieth century. For references, Kelly (1991)\cite{Kelly3} is remarkable. Books on social choice (from a number of different perspectives) include Sen (1970)\cite{Sen}, Fishburn (1973)\cite{Fish}, Feldman (1980)\cite{Feld}, Riker (1982)\cite{Riker}, Schofield (1985)\cite{Sch}, Campbell (1992)\cite{Cam}, Shepsle and Bonchek (1997)\cite{SB}, Austen-Smith and Banks (2000)\cite{AB}, and Arrow, Sen, and Suzumura, (2002)\cite{ASS}. Books on different aspects of voting include Straffin (1980)\cite{St}, Nurmi (1987)\cite{Nurmi}, Saari (1994)\cite{Saari2}, Felsenthal and Machover (1998)\cite{FM}, Taylor and Zwicker (1999)\cite{TZ}, and Saari (2001)\cite{Saari3}. An important recent survey is Brams and Fishburn (2002)\cite{BF}(see\cite{Taylor}, pp. 19,20).

Originating in economics and political science, social choice theory has since found its place as one of the fundamental tools for the study of multi-agent systems. The reasons for this development are clear: if we view a multi-agent system as a “society” of autonomous software agents, each of which has different objectives, is endowed with different capabilities, and possesses different information, then we require clearly defined and well-understood mechanisms for aggregating their views so as to be able to make collective decisions in such a multi-agent system(see\cite{BCE}, page:3). In fact, a burgeoning area—Computational Social Choice—aims to address problems in computational aspects of information/preference representation and aggregation in multi-agent scenarios(see\cite{X},\cite{CELM}and\cite{BCE}).

\end{document}